\documentclass[reqno,12pt,a4paper]{amsart}

\usepackage{amssymb}
\usepackage{url}
\usepackage{hyperref}

\allowdisplaybreaks[4]

\usepackage{mathrsfs}
\let\mathcal\mathscr
\usepackage[mathscr]{eucal}

\let\phi=\varphi
\let\kappa=\varkappa

\DeclareMathOperator{\sym}{sym}

\newcommand*{\sd}[2]{\{\,#1\mid#2\,\}}

\newcommand*{\abs}[1]{\left|#1\right|}
\newcommand*{\Ev}{\mathbf{E}}
\newcommand*{\Sv}{\mathbf{S}}

\newcommand*{\Int}{\mathbf{I}}

\theoremstyle{theorem}
\newtheorem{proposition}{Proposition}

\newtheorem{theorem}{Theorem}
\newtheorem{lemma}{Lemma}
\theoremstyle{definition}

\theoremstyle{remark}
\newtheorem{remark}{Remark}

\usepackage[all]{xy}
\SelectTips{cm}{}

\usepackage{mathrsfs}
\let\mathcal\mathscr
\usepackage[mathscr]{eucal}

\newcommand{\cprime}{\/{\mathsurround=0pt$'$}}

\begin{document}
\date{\today} \title[2D reductions and their nonlocal symmetries]{2D
  reductions of the equation \large
  $\mathbf{u_{yy} = u_{tx} + u_yu_{xx} - u_xu_{xy}}$\normalsize\\ and their
  nonlocal symmetries}

\author{P.~Holba} \address{Mathematical
  Institute, Silesian University in Opava, Na Rybn\'{\i}\v{c}ku 1, 746 01
  Opava, Czech Republic} \email{M160016@math.slu.cz}
\author{I.S.~Krasil{\cprime}shchik} \address{Independent University of Moscow,
  B. Vlasevsky 11, 119002 Moscow, Russia
  % Russian State University for
  % Humanities, Miusskaya sq. 6, Moscow, GSP-3, 125993, Russia
  \& Trapeznikov
  Institute of Control Sciences, 65 Profsoyuznaya street, Moscow 117997,
  Russia} \email{josephkra@gmail.com} \author{O.I.~Morozov} \address{Faculty
  of Applied Mathematics, AGH University of Science and Technology,
  Al. Mickiewicza 30, Krak\'ow 30-059, Poland}
\email{morozov{\symbol{64}}agh.edu.pl} \author{P.~Voj{\v{c}}{\'{a}}k}
\address{Mathematical Institute, Silesian University in Opava, Na
  Rybn\'{\i}\v{c}ku 1, 746 01 Opava, Czech Republic}
\email{Petr.Vojcak@math.slu.cz}

\date{\today}

\begin{abstract}
  We consider the 3D equation $u_{yy} = u_{tx} + u_yu_{xx} - u_xu_{xy}$ and
  its 2D reductions: (1) $u_{yy} = (u_y+y)u_{xx}-u_xu_{xy}-2$ (which is
  equivalent to the Gibbons-Tsarev equation) and (2)
  $u_{yy} = (u_y+2x)u_{xx} + (y-u_x)u_{xy} -u_x$. Using reduction of the known
  Lax pair for the 3D equation, we describe nonlocal symmetries of~(1) and~(2)
  and show that the Lie algebras of these symmetries are isomorphic to the
  Witt algebra.
\end{abstract}



\maketitle

\tableofcontents
\newpage

\section*{Introduction}

The equation
\begin{equation}
  \label{eq:1}
  u_{yy} = u_{tx} + u_yu_{xx} - u_xu_{xy}
\end{equation}
belongs to the class of linearly degenerate integrable
equations~\cite{Fer-Moss-2015} and was, as far as we know, introduced first
in~\cite{Mikhalev-1992} and independently
in~\cite{Pavlov-2003,Dunajski-2004}. In~\cite{B-K-M-V-2014}, we described all
two-dimensional symmetry reductions of this equation. All these reductions are
either linearizable or exactly solvable, except for the two ones:
\begin{equation}
  \label{eq:2}
  u_{yy} = (u_y+y)u_{xx}-u_xu_{xy}-2
\end{equation}
and
\begin{equation}
  \label{eq:3}
  u_{yy} = (u_y+2x)u_{xx} + (y-u_x)u_{xy} -u_x.
\end{equation}
Equation~\eqref{eq:2} is reduced to the Gibbons-Tsarev
equation~\cite{Gibbons-Tsarev-1996} by the simple transformation
$u \mapsto u-y^2/2$.

Equation~\eqref{eq:1} admits the Lax pair
\begin{equation}
  \label{eq:4}
  \begin{array}{lcl}
    w_{t}&=&(\lambda^2-\lambda u_x-u_y)w_x,\\
    w_{y}&=&(\lambda-u_x)w_{x}
  \end{array}
\end{equation}
with the non-removable spectral parameter~$\lambda$. In~\cite{B-K-M-V-2015},
we, in particular, studied the behavior of this Lax pair under the symmetry
reduction.

In this paper, we use system~\eqref{eq:4} to describe nonlocal symmetries of
Equations~\eqref{eq:2} and~\eqref{eq:3} and prove that in both cases these
symmetries form the Lie algebra isomorphic to the Witt algebra
\begin{equation*}
  \mathfrak{W} = \sd{e_i=z^{i+1}\textstyle{\frac{\partial}{\partial z}} }
  {i\in \mathbb{Z}}.
\end{equation*}
In Section~\ref{sec:preliminaries}, we recall some necessary results obtained
in the previous research and introduce the notions and constructions needed
for the subsequent exposition. Section~\ref{sec:main-result} is devoted to the
proofs of the basic results. We discuss in detail Equation~\eqref{eq:2} and
briefly repeat the main step for the second equation, because the reasoning is
quite similar in both cases. The obtained results are discussed in
Section~\ref{sec:discussion} together with further perspectives.

\section{Preliminaries}
\label{sec:preliminaries}

We recall here basic facts from the theory of nonlocal symmetries
(see~\cite{Kras-Vin-Trends-1989}) and previous results on
Equation~\eqref{eq:1} and its reductions, see~\cite{B-K-M-V-2014,
  B-K-M-V-2015}.

\subsection{Basics}
\label{sec:basics}

Consider a differential equation\footnote{We do not distinguish between scalar
and multi-component systems.}
\begin{equation}\label{eq:5}
  F^\alpha\left(x,\dots,\frac{\partial^{\abs{\sigma}}u}{\partial
      x^\sigma},\dots\right) = 0,\qquad \alpha=1,\dots,r,
\end{equation}
of order~$k$ in unknowns~$u=(u^1,\dots,u^m)$, where~$u^j=u^j(x)$,
$x=(x^1,\dots,x^n)$. To the infinite prolongation of~\eqref{eq:5} there
corresponds a locus $\mathcal{E}\subset J^\infty(\pi)$ in the space of
infinite jets of the trivial bundle
$\pi\colon \mathbb{R}^m\times\mathbb{R}^n \to \mathbb{R}^n$,
see~\cite{AMS-book}. Consider jet coordinates~$u_\sigma^j$ on~$J^\infty(\pi)$,
$j=1,\dots,m$, where~$\sigma=i_1\dots i_s$ is a symmetric multi-index,
$i_\alpha=1,\dots,n$, and choose a set $\Int=\{u_\sigma^j\}$ of internal
variables on~$\mathcal{E}$. Denote by
\begin{equation*}
  D_{x^i} = \frac{\partial}{\partial x^i} + \sum_{\Int}u_{\sigma
    i}\frac{\partial}{\partial u_\sigma^j},\qquad i=1,\dots,n,
\end{equation*}
the operators of total derivatives on~$\mathcal{E}$. For any function~$f$
on~$J^\infty(\pi)$ we denote by~$\ell_f$ the restriction of its linearization
$\sum_{j,\sigma} \partial f/\partial u_\sigma^j D_\sigma$ to~$\mathcal{E}$,
where $D_\sigma = D_{x^{i_1}}\circ\dots\circ D_{x^{i_s}}$. A (local) symmetry
of~$\mathcal{E}$ is an evolutionary vector field
\begin{equation*}
  \Ev_\phi = \sum_{\Int}D_\sigma(\phi^j)\frac{\partial}{\partial u_\sigma^j},
\end{equation*}
where the generating function $\phi=(\phi^1,\dots,\phi^m)$ satisfies the
equation $\ell_F(\phi)=0$ and $F=(F^1,\dots,F^r)$ determines
Equation~\eqref{eq:5}. We identify symmetries with their generating functions
and denote the Lie algebra of symmetries by~$\sym\mathcal{E}$.

A (differential) covering over~$\mathcal{E}$ is a bundle
$\tau\colon \tilde{\mathcal{E}}= \mathcal{E}\times \mathbb{R}^l\to
\mathcal{E}$, $l=1,2,\dots,\infty$, endowed with vector fields
\begin{equation*}
  \tilde{D}_{x^i} = D_{x^i} + \sum_\alpha X_i^\alpha\frac{\partial}{\partial
    w^\alpha},\qquad i=1,\dots,n,
\end{equation*}
that pair-wise commute, where~$w^\alpha$ are coordinates in~$\mathbb{R}^l$
called nonlocal variables. Equivalently, $\tilde{\mathcal{E}}$ may be regarded
as an overdetermined system
\begin{equation*}
  w_{x^i}^\alpha = X_i^\alpha,\qquad i=1,\dots,n,\quad \alpha = 1,\dots,l,
\end{equation*}
whose compatibility conditions are consequences of~$\mathcal{E}$. For a linear
differential operator
$\Delta = \left(\sum_\sigma a_{ij}^\sigma D_\sigma\right)$ on~$\mathcal{E}$ we
denote by
$\tilde{\Delta} = \left(\sum_\sigma a_{ij}^\sigma \tilde{D}_\sigma\right)$ its
lift to~$\tilde{\mathcal{E}}$.

Coverings~$\tau_1$ and~$\tau_2$ are said to be equivalent if there exists a
diffeomorphism $g\colon \tilde{\mathcal{E}}_1 \to \tilde{\mathcal{E}}_2$ such
that $\tau_2\circ g = \tau_1$ and which takes~$\tilde{\mathcal{C}}_1$
to~$\tilde{\mathcal{C}}_2$, where~$\tilde{\mathcal{C}}_i$ is the span of the
fields~$\tilde{D}_{x^i}$ on~$\tilde{\mathcal{E}}_i$.

A nonlocal symmetry of~$\mathcal{E}$ in the covering~$\tau$ is a vector field
\begin{equation*}
  \Sv_\Phi = \sum_\Int\tilde{D}_\sigma(\phi^j)\frac{\partial}{\partial
    u_\sigma^j} +
  \sum_\alpha \psi^\alpha\frac{\partial}{\partial w^\alpha},
\end{equation*}
where $\phi=(\phi^1,\dots,\phi^m)$ and~$\psi^\alpha$ are functions
on~$\tilde{\mathcal{E}}$ satisfying the system
\begin{align*}
  &\tilde{\ell}_F(\phi)=0,\\
  &\tilde{D}_{x^i}(\psi^\alpha) = \tilde{\ell}_{X_i^\alpha}(\phi) + \sum_\beta
    \frac{\partial X_i^\alpha}{\partial w^\beta}
\end{align*}
for all~$i$ and~$\alpha$. Any nonlocal symmetry is identified with the
collection $\Phi = (\phi,\psi^1,\dots,\psi^\alpha,\dots)$. A symmetry is
called a lift of a local symmetry if the component~$\phi$ is a function
on~$\mathcal{E}$; it is called invisible if~$\phi=0$. Nonlocal symmetries form
a Lie algebra denoted by~$\sym_\tau\mathcal{E}$.

\subsection{Coverings over Equation~\eqref{eq:1} and symmetries}
\label{sec:cover-over-equat}

Let us assume that $w=w(\lambda)$ in Equation~\eqref{eq:4} and consider the
expansion
\begin{equation*}
  w=\sum_{i=-\infty}^{+\infty}\lambda^{-i} w_i.
\end{equation*}
This leads to the infinite-dimensional covering
\begin{equation}\label{eq:6}
  \begin{array}{rcl}
    w_{i,x} &=& w_{i-1,y} + u_x w_{i-1,x} ,\\
    w_{i,y} &=& w_{i-1,t} + u_y w_{i-1,x},
  \end{array}
\end{equation}
where $i\in\mathbb{Z}$.

Recall now that the space~$\sym\mathcal{E}$ for Equation~\eqref{eq:1} is
spanned by the functions
\begin{gather*}
   \theta_1=2x-yu_x,\qquad   \theta_2=3u-2xu_x-yu_y,\\
   \theta_3(T)=Tu_y+T'(yu_x-x)-\frac{1}{2}T''y^2, \qquad   \theta_4(T)=Tu_x-T'y,\\
   \theta_5(T)=Tu_t+T'(xu_x+yu_y-u)+\frac{1}{2}T''(y^2u_x-2xy)-
  \frac{1}{6}T'''y^3,  \qquad
  \theta_6(T)=T,
\end{gather*}
where~$T$ is a function in~$t$ and `prime' denotes the $t$-derivatives. In
what follows, we shall need the following

\begin{proposition}\label{prop:cover-over-equat-1}
  The symmetries $\theta_1$\textup{,} $\theta_2$\textup{,} and
  $\theta_5 = \theta_5(1)= u_t$ can be lifted to the covering~\eqref{eq:6}.
\end{proposition}

\begin{proof}
  Let us set
  \begin{align*}
    \theta_1^i&= -y w_{i,x} + (i + 2) w_{i-1},\\
    \theta_2^i&= -2 x w_{i,x} - y w_{i,y} + (i + 3) w_i,\\
    \theta_5^i&= w_{i,t}.
  \end{align*}
  Then
  \begin{equation*}
    \Theta_j=(\theta_j,\dots,\theta_j^i,\dots),\qquad j=1,2,5,\quad
    i\in\mathbb{Z}, 
  \end{equation*}
  are the desired lifts.
\end{proof}

\subsection{Reductions}
\label{sec:reductions}
Using Proposition~\ref{prop:cover-over-equat-1}, we now state the following
\begin{proposition}
  \label{prop:reductions-2}
  Reduction of~\eqref{eq:6} with respect to the symmetry $\Theta_5+\Theta_1$
  leads to the covering
  \begin{equation}
    \label{eq:7}
    \begin{array}{rcl}
      w_{i,x} &=& w_{i-1,y} + u_x w_{i-1,x} ,\\
      w_{i,y} &=& (u_y + y)w_{i-1,x} - (i + 1)w_{i-2}
    \end{array}
  \end{equation}
  over Equation~\eqref{eq:2}\textup{,} while reduction with respect to
  $\Theta_5 + \Theta_2$ leads to the covering
  \begin{equation}
    \label{eq:8}
    \begin{array}{rcl}
      w_{i,x} &=& w_{i-1,y} + u_x w_{i-1,x} ,\\
      w_{i,y} &=& (u_y + 2x)w_{i-1,x} + yw_{i-1,y} - (i + 2)w_{i-1}
    \end{array}
  \end{equation}
  over Equation~\eqref{eq:3}.
\end{proposition}

Let us finally describe the algebras of local symmetries for
Equations~\eqref{eq:2} and~\eqref{eq:3}. Direct computations show that
$\sym\mathcal{E}_1$ for the first equation is spanned by the functions
\begin{equation}\label{eq:12}
  \phi_{-4} = 1,\quad \phi_{-3} = u_x,\quad \phi_{-2}= -u_y - y,\quad
  \phi_{-1} = -2 x + yu_x,\quad
  \phi_0 = 4 u - 3 xu_x - 2 yu_y,
\end{equation}
while in the second case we have the following generators of the symmetry
algebra $\sym\mathcal{E}_2$:
\begin{equation*}
  \gamma_{-3} = 1,\quad \gamma_{-2}=-y-\frac{1}{2}u_x,\quad \gamma_{-1} =  y^2  - 2
  x + 2 yu_x - 2 u_y,\quad \gamma_0= 3 u - 2x u_x - yu_y.
\end{equation*}
It is also easily seen that
\begin{equation*}
  \left.\begin{array}{l}
    \text{for Equation~\eqref{eq:2}}\\
     \protect{[\phi_i,\phi_j]= }
    \begin{cases}
      (j-i)\phi_{i+j}&\text{if } i+j\geq -4,\\
      0&\text{otherwise},
    \end{cases}
  \end{array}\qquad\right|\qquad
  \begin{array}{l}
    \text{for Equation~\eqref{eq:3}}\\
    \protect{[\gamma_i,\gamma_j]= }
    \begin{cases}
      (j-i)\gamma_{i+j}&\text{if } i+j\geq -3,\\
      0&\text{otherwise}.
    \end{cases}
  \end{array}
\end{equation*}
Denote by~$\mathfrak{W}_0\subset\mathfrak{W}$ the subalgebra in the Witt
algebra spanned by the fields~$e_i$ with~$i\leq 0$ and by~$\mathfrak{W}_k$ its
ideal spanned by~$e_i$, $i\leq k < 0$.

\begin{proposition}
  \label{prop:reductions-1}
  For Equation~\eqref{eq:2}\textup{,} one has
  $\sym\mathcal{E}_1=\mathfrak{W}_0/\mathfrak{W}_5$\textup{,} while
  $\sym\mathcal{E}_2=\mathfrak{W}_0/\mathfrak{W}_4$ for Equa\-ti\-on~\eqref{eq:3}.
\end{proposition}

\subsection{The coverings $\tau^p$ and $\sigma^p$}
\label{sec:coverings-taup}

Consider the covering~\eqref{eq:7} over Equation~\eqref{eq:2}. Choose an
integer~$p\in\mathbb{Z}$ and assume that~$w_i=0$ for all~$i<p$. To proceed
further, it is convenient to relabel nonlocal variables by
setting~$w_{p+i}=r_{i-4}^p$. Then~$r_i^p=0$ for~$i<-4$ and
\begin{align*}
  &r_{-4,x}^p=0,&&r_{-4,y}^p=0;\\
  &r_{-3,x}^p=0,&&r_{-3,y}^p=0,
\end{align*}
and without loss of generality one can set~$r_{-4}^p=1$, $r_{-3}^p=0$. Then
\begin{equation*}
  r_{-2}^p=-(p + 3)y,\qquad r_{-1}^p = -(p + 3)x,\qquad r_0^p = -(p + 3)u +
  \frac{1}{2}(p + 3)(p + 4)y^2,
\end{equation*}
while
\begin{equation}
  \label{eq:9}
  \begin{array}{rcl}
    r_{i,x}^p&=&r_{i-1,y}^p + u_xr_{i-1,x}^p,\\[3pt]
    r_{i,y}^p&=&(u_y+y)r_{i-1,x}^p-(p+i+5)r_{i-2}^p,
  \end{array}
\end{equation}
for all~$i\geq 1$. We denote this covering by~$\tau^p$.

In a similar way, for Equation~\eqref{eq:3} and its covering~\eqref{eq:8} we
have~$r_i^p=0$ for~$i<-3$, $r_{-3}^p = 1$, and
\begin{equation*}
  r_{-2}^p = -(p + 3)y,\quad r_{-1}^p = -(p + 3)x +
  \frac{1}{2}(p + 3)^2y^2,\quad r_0^p = - (p + 3)u + (p + 3)^2 xy -
  \frac{1}{6} (p + 3)^3 y^3,
\end{equation*}
while
\begin{equation}
  \label{eq:10}
  \begin{array}{rcl}
    r_{i,x}^p&=& r_{i-1,y}^p + u_x r_{i-1,x}\\
    r_{i,y}^p&=& (u_y + 2x)r_{i-1,x}^p + yr_{i-1,y}^p - (p + i + 5)r_{i-1}^p
  \end{array}
\end{equation}
for~$i\geq1$.
This covering will be denoted by $\sigma^p$.

\section{The main result}
\label{sec:main-result}

We prove the main result of the paper in this section, which states that in
coverings~$\tau^p$, naturally associated with~\eqref{eq:7} and~\eqref{eq:8}
the algebras of nonlocal symmetries for both Equations~\eqref{eq:2}
and~\eqref{eq:3} are isomorphic to the Witt algebra~$\mathfrak{W}$. The case
of Equation~\eqref{eq:2} is considered in detail, while for
Equation~\eqref{eq:3} we provide a sketch of proofs only.

\subsection{Equivalence of the coverings $\tau^p$}
\label{sec:equiv-cover-taup}

The first step of the proof is to establish the equivalence of
different~$\tau^p$ to each other.

\begin{proposition}
  \label{prop:equiv-cover-taup-1}
  The covering~$\tau^p$ is equivalent to~$\tau^{-4}$ for any~$p$.
\end{proposition}

\begin{proof}
  We distinguish the two cases: $p\neq -3$ and $p=3$.

  \textbf{The case~$p\neq -3$.} Introduce the notation~$s_i=r_i^{-4}$ and
  consider the vector field
  \begin{equation*}
    \mathcal{X}=x\frac{\partial}{\partial y} + 2u\frac{\partial}{\partial x} +
    3s_1\frac{\partial}{\partial u} + \sum_{i\geq
      1}(i+3)s_{i+1}\frac{\partial}{\partial s_i}.
  \end{equation*}
  Let us define the quantities~$Q_{k,j}$, $k\geq0$, by
  \begin{equation*}
    Q_{k,0} = \frac{1}{(k+2)!}y^{k+2},\qquad Q_{k,j+1} =
    \frac{1}{j}\mathcal{X}(Q_{k,j})
  \end{equation*}
  and assume~$Q_{k,j}=0$ for~$j<0$. Set
  \begin{equation*}
    d_i = \sum_{k\geq 0}(-1)^k\frac{(p+k+4)!}{(p+4)!}Q_{k,i-2k}.
  \end{equation*}
  Then the transformation
  \begin{equation*}
    r_i^p = -(p+3)(s_i-(p+4)d_i),\qquad i\geq1,
  \end{equation*}
  is the desired equivalence.

  \textbf{The case~$p= -3$.}  The first three pairs of defining equations in
  this case are
  \begin{equation*}
    r_{i,x}=0,\quad r_{i,y}=0,\qquad i=-3,-2,-1,
  \end{equation*}
  while all the rest ones do not contain the variable~$r_{-3}$. Thus we
  that~$r_{-3}$ is a constant whose value does not influence the subsequent
  computations. So, we may set~$r_{-3}=0$ and this reduces the case under
  consideration to~$p>-3$.
\end{proof}

\subsection{The Lie algebra of  nonlocal symmetries}
\label{sec:lie-algebra-tructure}

Let~$\Phi=(\phi,\phi^1,\dots,\phi^i,\dots)$ be a nonlocal symmetry of
Equation~\eqref{eq:2} in the covering~\eqref{eq:9}. The defining equations for
the components of~$\Phi$ are
\begin{align}\nonumber
  \tilde{\ell}_F(\phi)&\equiv \tilde{D}_y^2(\phi) - (u_y +
  y)\tilde{D}_x^2(\phi) - u_{xx}\tilde{D}_y(\phi)  + u_x
  \tilde{D}_x\tilde{D}_y(\phi)  + u_{xy}\tilde{D}_x(\phi)=0, \\\label{eq:11}
  \tilde{D}_x(\phi^i)&=\tilde{D}_y(\phi^{i-1}) + u_x\tilde{D}_x(\phi^{i-1}) +
  r_{i-1,x}^p\tilde{D}_x(\phi), \\ \nonumber
  \tilde{D}_y(\phi^i)&=(u_y + y)\tilde{D}_x(\phi^{i-1}) - (p + i +
  1)\phi^{i-2} + r_{i-1,x}^p\tilde{D}_y(\phi).
\end{align}
We need a number of auxiliary results to describe the
algebra~$\sym_{\tau^p}\mathcal{E}_1$. As a first step, it is convenient to
assign weights to the internal coordinates in~$\tilde{\mathcal{E}}_1$ in such a
way that all polynomial objects become homogeneous with respect to these
weights. Let us set
\begin{equation*}
  \abs{x} = 3,\qquad \abs{y} = 2,\qquad \abs{u} = 4.
\end{equation*}
Then
\begin{equation*}
  \abs{u_x} = \abs{u} - \abs{x} = 1,\qquad \abs{u_y} = \abs{u} - \abs{y} = 2,
\end{equation*}
etc., and
\begin{equation*}
  \abs{r_i^p} = i+4,\qquad i\geq1.
\end{equation*}
The weight of a monomial is the sum of weights of its factors and the weight
of a vector field $R\partial/\partial \rho$ is $\abs{R}-\abs{\rho}$. In
particular, one has~$\abs{\phi_k}=k$, $k=-4,\dots,0$ for the local symmetries
presented in~\eqref{eq:12}.

\begin{lemma}
  \label{lem:lie-algebra-nonlocal}
  The local symmetry~$\phi_{-1}$ can be lifted to the covering~$\tau^{-4}$.
\end{lemma}

\begin{proof}
  Let us set $\Phi_{-1}=(\phi_{-1},\phi_{-1}^1,\dots,\phi_{-1}^i,\dots)$,
  where $\phi_{-1}^i=yr_{i,x}^{-4}-(i+2)r_{i-1}^{-4}$. It is straightforward
  to check that Equations~\eqref{eq:11} are satisfied for~$p=-4$.
\end{proof}

\begin{lemma}
  \label{lem:lie-algebra-nonlocal-1}
  The local symmetry~$\phi_{-2}$ can be lifted to the
  covering~$\tau^{-5}$.
\end{lemma}

\begin{proof}
  We set $\Phi_{-2} =
  (\phi_{-2},-r_{1,y}^{-5},\dots,-r_{i,y}^{-5},\dots)$. Then~\eqref{eq:11} are
  fulfilled in an obvious way for~$p=-5$.
\end{proof}

\begin{lemma}
  \label{lem:lie-algebra-nonlocal-2}
  There exists a nonlocal symmetry of weight~$2$ in the covering~$\tau^{-1}$.
\end{lemma}

\begin{proof}
  Let us set $\Phi_2=(\phi_2,\phi_2^1,\dots,\phi_2^i,\dots)$, where
  \begin{equation*}
    \phi_2=
    -\frac{1}{2}\,\left(8 r_2^{-1}+42 y (3 y^2-2 u)+(80 x y-7 r_1^{-1}) u_x
    +2 (7 y^2+6 u) u_y-50   x^2-96 y^3\right)
  \end{equation*}
  and
  \begin{equation*}
    \phi_2^i=
    -\frac{1}{2}\,\left(
    (80 x y-7 r_1^{-1}) r_{i, x}^{-1}+2 (7 y^2+6 u) r_{i, y}^{-1}
    -2 (i+8) r_{i+2}^{-1}-8 y (i+6) r_{i}^{-1}-10 (i+5) x r_{i-1}^{-1}
    \right)
  \end{equation*}
  for all $i\geq1$.
\end{proof}

\begin{proposition}
  \label{prop:lie-algebra-nonlocal}
  The symmetries~$\Phi_{-2}$\textup{,} $\Phi_{-1}$\textup{,} and~$\Phi_2$
  exist in any covering~$\tau^p$.
\end{proposition}

\begin{proof}
  A direct consequence of
  Lemmas~\ref{lem:lie-algebra-nonlocal}--~\ref{lem:lie-algebra-nonlocal-2} and
  Proposition~\ref{prop:equiv-cover-taup-1}.
\end{proof}

\begin{theorem}
  \label{thm:lie-algebra-nonlocal}
  The symmetries~$\Phi_{-2}$\textup{,} $\Phi_{-1}$\textup{,} and~$\Phi_{2}$
  generate the entire Lie algebra~$\sym_{\tau^p}\mathcal{E}_1$ which is
  isomorphic to the Witt algebra~$\mathfrak{W}$\textup{,} i.e.\textup{,} there
  exists a basis~$\Phi_k$\textup{,} $k\in\mathbb{Z}$\textup{,} such that
  \begin{equation*}
    [\Phi_k,\Phi_l] =(l-k)\Phi_{k+l}
  \end{equation*}
  for all~$k$\textup{,} $l\in\mathbb{Z}$.
\end{theorem}

\begin{proof}
  Let set
  \begin{equation*}
    \Phi_0=4[\Phi_{-2},\Phi_2],\qquad \Phi_1=3[\Phi_{-1},\Phi_2]
  \end{equation*}
  and by induction
  \begin{equation*}
    \Phi_{k+1} = (k-1)[\Phi_1,\Phi_k],\qquad \Phi_{-k-1} =
    (1-k)[\Phi_{-1},\Phi_{-k}]
  \end{equation*}
  for all $k\geq 2$.
\end{proof}

\begin{remark}
  The symmetries~$\Phi_k$ are invisible for~$k<-4$.
\end{remark}

\begin{remark}
  The symmetries $\Phi_{-4}$, $\Phi_{-3}$, and $\Phi_0$ are lifts of local
  symmetries $\phi_{-4}$, $\phi_{-3}$, $\phi_0$, respectively.
\end{remark}

\subsection{Explicit formulas}
\label{sec:explicit-formulas}

To obtain explicit formulas for nonlocal symmetries in~$\tau^p$ for
arbitrary\footnote{The case~$p=-3$ is special, but, as it was indicated above,
  is reduced to the case~$p=-2$.} value of~$p$ consider the fields
\begin{equation*}
  \mathcal{Y}_m = \sum_{i=m}^\infty(i-m+1)r_{i-3}^p\frac{\partial}{\partial
    r_{i-4}^p}
\end{equation*}
on the space of~$\tau^p$, where $m=1\dots,4$, and the quantities~$P_{i,j}^m$
defined by induction as follows:
\begin{equation*}
  P_{i,0}^{(m)} = \frac{1}{(i+2)!}\left(r_{m-4}^p\right)^{i+2},\qquad
  P_{i,j+1}^{(m)} = \frac{1}{j+1}\mathcal{Y}_m(P_{i,j}^{(m)}),
\end{equation*}
where~$i$, $j\geq 0$, $m=1,\dots,4$. We set~$P_{i,j}^{(m)}=0$ if at least one
of the subscripts is~$<0$. In terms of these quantities, the lift
of~$\phi_{-2}$ to~$\tau^p$ acquires the form
\begin{equation*}
  \Phi_{-2}^p = (\phi_{-2},\phi_{-2}^{1,p},\dots,\phi_{-2}^{i,p},\dots),
\end{equation*}
where
\begin{equation*}
  \phi_{-2}^{i,p} = -r_{i,y}^p - (p+5)\left(r_{i-2}^p + \sum_{j=1}^\infty
    \left(-\frac{1}{p+3}\right)^j\cdot\prod_{l=0}^{j-1}((p+3)l-2)\cdot
    P_{j-1,i-2j}^{(2)}\right).
\end{equation*}
The lift
\begin{equation*}
  \Phi_{-1}^p = (\phi_{-1},\phi_{-1}^{1,p},\dots,\phi_{-1}^{i,p},\dots)
\end{equation*}
of~$\phi_{-1}$ is given by
\begin{equation*}
  \phi_{-1}^{i,p} = yr_{i,x}^p - (i+2)r_{i-1}^p + (p+4)\sum_{j=1}^\infty
  \left(-\frac{1}{p+3}\right)^j\cdot  \prod_{l=0}^{j-1}((p+3)l-1)\cdot
  P_{j-1,i-2j+1}^{(2)}.
\end{equation*}
The nonlocal symmetry~$\Phi_2$, when passing to~$\tau^p$, acquires the form
\begin{equation*}
  \Phi_{2}^p = (\phi_{2}^p,\phi_{2}^{1,p},\dots,\phi_{2}^{i,p},\dots)
\end{equation*}
with
\begin{align*}
  \phi_{2}^p&=
  -\frac{1}{p+3}\,\left(
  8r_2^p + 2y(4p+25)r_0^p
      + \Big((p+3)(7p+47)xy - 7r_1^p\Big)u_x\right)
  \\
  &
  -\Big((7y^2+6y)u_y - (4p+29)x^2 - \frac{1}{3}y^3(8p^2 + 87p
  +223)\Big)
  \end{align*}
and
\begin{align*}
  \phi_{2}^{i,p}&=
  -\frac{1}{p+3}\,\Big((p+3)(7p+47)xy - 7r_1^p\Big)r_{i,x}^p
  -(7y^2 +  6u)r_{i,y}^p
  \\
  &    +(i+8)r_{i+2}^p + 2y(p+2i+13)r_i^p + x(3p+5i+28)r_{i-1}^p
    \\
  &
  +\frac{p+1}{p+3}\,\left(2P_{0,i-2}^{(4)} +
    \sum_{j=1}^\infty\left(-\frac{1}{p+3}\right)^j 
    \cdot \prod_{l=0}^j(l(p+3)+2)\cdot P_{j,i-2j+2}^{(2)}\right).
\end{align*}
Explicit formulas for other nonlocal symmetries can be obtained using
commutator relations from the proof of
Theorem~\ref{thm:lie-algebra-nonlocal}. For example, the symmetry~$\Phi_1$
that generates the positive part of~$\sym_{\tau^p}\mathcal{E}_1$ is
\begin{equation*}
  \Phi_1=(\phi_1^p,\phi_1^{1,p},\dots,\phi_1^{i,p},\dots)
\end{equation*}
with
\begin{equation*}
  \phi_1^p=-\frac{6}{p+3}\,r_1^p -\big((4y^2+5u)u_x + 4xu_y - 2(3p+16)xy\big)
\end{equation*}
and
\begin{align*}
  \phi_1^{i,p}&=
  -\big((4y^2+5u)r_{i,x}^p + 4xr_{i,y}^p - (i+6)r_{i+1}^p -
  y(2p+3i+16)r_{i-1}^p\big)
  \\
  &
  +\frac{p+2}{p+3}\,\left(P_{0,i-1}^{(3)} + \sum_{j=1}^\infty
    \left(-\frac{1}{p+3}\right)^j\cdot\prod_{l=0}^j(l(p+3)+1) \cdot
    P_{j,i-2j+1}^{(2)}\right).
\end{align*}

\subsection{Equation \eqref{eq:3}}
\label{sec:equation-eq:3}

The results on this equation and their proofs are almost identical to those on
Equation~\eqref{eq:2}. So, we confine ourselves with the final description
of~$\sym_{\sigma^p}\mathcal{E}_2$. Introduce the weights
\begin{equation*}
  \abs{x} = 2,\qquad \abs{y} = 1,\qquad \abs{u} = 3.
\end{equation*}
Consequently, the nonlocal variables in the covering~$\sigma^p$ acquire the
weights~$\abs{r_i^p} = i+3$.

Then we have

\begin{theorem}
  \label{thm:equation-eqrefeq:3}
  The symmetries~$\gamma_{-2}$ and~$\gamma_{-1}$ can be lifted to
  symmetries~$\Gamma_{-2}$ and~$\Gamma_{-1}$ in any covering~$\sigma^p$ over
  Equation~\eqref{eq:3}. In addition\textup{,} there exists a nonlocal
  symmetry~$\Gamma_2$ of weight~$2$. These three symmetries generate the
  entire Lie algebra~$\sym_{\sigma^p}\mathcal{E}_2$ which is isomorphic to the
  Witt algebra~$\mathfrak{W}$.
\end{theorem}

\begin{proof}
  The symmetry $\Gamma_{-2} =
  (\gamma_{-2},\gamma_{-2}^{1,p},\dots,\gamma_{-2}^{i,p},\dots)$ is given
  by~$\gamma_{-2} = -y-u_x/2$ and
  \begin{equation*}
    \gamma_{-2}^{i,p} = -\frac{1}{2}\left(r_{i,x}^p + (p+5)\left(r_{i-2}^p +
        \sum_{j=1}^\infty \left(-\frac{1}{p+3}\right)^j \cdot \prod_{l=0}^{j-1}
          (l(p+3) - 2)\cdot P_{j-1,i-j}^{(1)}\right)\right).
  \end{equation*}
  For $\Gamma_{-1} =
  (\gamma_{-1},\gamma_{-1}^{1,p},\dots,\gamma_{-1}^{i,p},\dots)$ we have
  $\gamma_{-1}= y^2 - 2 x + 2 yu_x - 2 u_y$ and
  \begin{equation*}
    \gamma_{-1}^{i,p} = 2\left(yr_{i,x}^p - r_{i,y}^p - (p+4)\left(r_{i-1}^p +
        \sum_{j=1}^\infty \left(-\frac{1}{p+3}\right)^j \cdot
        \prod_{l=0}^{j-1} (l(p+3)-1)\cdot P_{j-1,i-j+1}^{(1)}\right)\right).
  \end{equation*}
  Finally, we have
  \begin{align*}
    \gamma_2^p &= 7r_2^p + \big((7p + 37)y-6u_x\big)r_1^p\\
    &+(p + 3) \left( (6p + 31)yu + (3p + 19)x^2 - (3p^2 + 18p + 13)xy^2 +
      \frac{1}{12}(3p^3+ 27p^2 + 81p + 89)y^4\right.\\
    &\left.  + \left(5u + 14xy + \frac{19}{3}y^3\right)u_y - \frac{1}{2}(7p^2+
      74p + 197)y^2 u - (7p + 37)xu - 4(2p +
      7)x^2y\right. \\
    &\left.+ \frac{1}{3}(7p^3+ 87p^2 + 333p + 397)xy^3 - \frac{1}{30}(7p^4+
      104p^3 + 558p^2 + 1296p + 1115)y^5\right)
  \end{align*}
  and
  \begin{equation*}
    \gamma_2^{i,p}=-(p+3)((i + 7)r_{i+2}^p + (p + 3i + 19)yr_{i+1}^p) +
    Ar_{i,x}^p + Br_{i,y}^p + Cr_i^p +D
  \end{equation*}
  for~$\Gamma_2 = (\gamma_2^p,\gamma_2^{1,p},\dots,\gamma_2^{i,p},\dots)$, where
  \begin{align*}
    A&=(p + 3)\left(
      (6p + 31)uy + (3p + 19)x^2 - (3p^2 + 18p + 13)xy^2 +
      \frac{1}{12}(3p^3+ 27p^2 + 81p + 89)y^4\right)\\
    &- 6r_1^p,\\
    B&= (p + 3)\left(5u + 14xy + \frac{19}{3}y^3\right),\\
    C&=- (p + 3)\Big( 2(p + 2i + 11)x + (p^2+ 9p + 5i + 33)y^2 \Big),\\
    \intertext{and}
    D&=-(p + 1)\left(2P_{0,i-1}^{(3)} + \sum_{j=1}^\infty
      \left(-\frac{1}{p+3}\right)^j\cdot
      \prod_{l=0}^j(l(p + 3) + 2)\cdot P_{j,i-j+3}^{(1)}
    \right)
  \end{align*}
  for~$i\geq 1$.
\end{proof}

\section{Discussion}
\label{sec:discussion}

We conclude with two remarks.

\begin{remark}
  \label{sec:discussion-1}
  Consider the covering~\eqref{eq:9} with $p=2$  and the generating series
  \begin{equation*}
    R=\sum_{i=1}^\infty\lambda^ir_i
  \end{equation*}
  of the corresponding nonlocal variables. The function~$R$ satisfies the
  system
  \begin{align*}
    R_x&=\lambda(R_y + u_xR_x),\\
    R_y&=\lambda(u_y+y)R_x - (\lambda^3R)_\lambda,
  \end{align*}
  which is equivalent to
  \begin{equation}\label{eq:14}
    \begin{array}{lcr}
    \tilde{R}_x&=\dfrac{\lambda^4\tilde{R}_\lambda}{\lambda^2(u_y+y) +
      \lambda u_x - 1},\\[10pt]
    \tilde{R}_y&=\dfrac{\lambda^3(1-\lambda
      u_x)\tilde{R}_\lambda}{\lambda^2(u_y+y) + \lambda u_x - 1},
  \end{array}
  \end{equation}
  where~$\tilde{R} = \lambda^3R$.  Let $\tilde{R}=\tilde{R}(x,y,\lambda)$ be a
  solution to \eqref{eq:14} such that $\tilde{R}_\lambda \not \equiv 0$. Then
  equation $\tilde{R}(x,y,\lambda) =\mathrm{const}$ defines $\lambda$ as a
  fuction of $x$ and $y$, that is, $\lambda = \psi(x,y)$,
  cf.~\cite{Pavlov-2009}.  Then
  \begin{equation*}
    \tilde{R}_\lambda\cdot\psi_x + \tilde{R}_x = 0,\quad
    \tilde{R}_\lambda\cdot\psi_y + \tilde{R}_y = 0
  \end{equation*}
  and~\eqref{eq:14} transforms to
    \begin{equation*}
    \begin{array}{lcr}
      \psi_x&=&\dfrac{\psi^4}{1-u_x\psi-(u_y+y)\psi^2},\\[10pt]
      \psi_y&=&\dfrac{\psi^3(1-u_x\psi)}{1-u_x\psi-(u_y+y)\psi^2},
    \end{array}
  \end{equation*}
  which, by the gauge transformation~$\psi\mapsto\psi^{-1}$, is equivalent to
  \begin{equation*}
    \begin{array}{lcr}
      \psi_x&=&-\dfrac{1}{\psi^2-u_x\psi-(u_y+y)},\\[10pt]
      \psi_y&=&\dfrac{u_x-\psi}{\psi^2-u_x\psi-(u_y+y)}.
    \end{array}
  \end{equation*}
  This is the covering obtained in~\cite{B-K-M-V-2015} by the direct reduction
  and coincides with the known covering over the Gibbons-Tsarev equation,
  see~\cite{Gibbons-Tsarev-1996}.
\end{remark}

\begin{remark}
  \label{sec:discussion-2}
  In~\cite{B-K-M-V-2014}, two other equations,
  \begin{equation*}
    u_y u_{xy} - u_x u_{yy} = e^y u_{xx}
  \end{equation*}
  and
  \begin{equation*}
    u_{yy} = (u_x + x)u_{xy} - u_y (u_{xx} + 2),
  \end{equation*}
  were obtained as symmetry reductions of the universal hierarchy and the 3D
  rdDym equations, respectively, were obtained. We plan study these equations
  by the methods similar to the used above in the forthcoming publications.
\end{remark}

\section*{Acknowledgments}

Computations were supported by the \textsc{Jets} software,~\cite{Jets}.

The first author (PH) was supported by the Specific Research grant SGS/6/2017
of the Silesian University in Opava.  The third author (OIM) is grateful to
the Polish Ministry of Science and Higher Education for financial support.


\begin{thebibliography}{99}
\bibitem{B-K-M-V-2014} H.~Baran, I.S.~Krasil{\cprime}shchik, O.I.~Morozov, and
  P.~Voj\v{c}\'{a}k, \emph{Symmetry reductions and exact solutions of Lax
    integrable $3$-dimensional systems}, J.\ of Nonlinear Math.\ Phys.,
  \textbf{21}, Number 4, 2014, 643--671.

\bibitem{B-K-M-V-2015} H.~Baran, I.S.~Krasil{\cprime}shchik, O.I.~Morozov, and
  P.~Voj\v{c}\'{a}k, \emph{Integrability properties of some equations obtained
    by symmetry reductions}, J.\ of Nonlinear Math.\ Phys., \textbf{22},
  Number 2, 2015, 210--232.

\bibitem{Jets} H.~Baran, M.~Marvan, \emph{Jets. A software for differential
    calculus on jet spaces and diffeties}.  \url{http://jets.math.slu.cz}.

\bibitem{AMS-book} A.V.~Bocharov et al., \emph{Symmetries of Differential
    Equations in Mathematical Physics and Natural Sciences}, edited by
  A.M.~Vinogradov and I.S.~Krasil{\cprime}shchik). Factorial Publ. House, 1997
  (in Russian). English translation: Amer.  Math. Soc., 1999.


\bibitem{Dunajski-2004} M.~Dunajski,
\emph{A class of Einstein--Weil spaces associated to an integrable system of hydrodynamic type},
J.\ Geom.\ Phys., \textbf{51} (2004), 126--137.

\bibitem{Fer-Moss-2015} E.V.~Ferapontov, J.~Moss, \emph{Linearly degenerate
    partial differential equations and quadratic line complexes}, Comm.\ in
  Anal.\ and Geom., \textbf{23} (2015) no.~1, 91--127.

\bibitem{Gibbons-Tsarev-1996} J.~Gibbons and S.P.~Tsarev, \emph{Reductions of
    the Benney equations}, Physics Letters A, \textbf{211}, Issue 1, 19--24,
  1996.

\bibitem{Kras-Vin-Trends-1989} I.S.~Krasil{\cprime}shchik and A.M.~Vinogradov,
  \emph{Nonlocal trends in the geometry of differential equations: symmetries,
    conservation laws, and B\"acklund transformations}, Acta Appl.\ Math.\
  \textbf{15} (1989) no.~1-2.  Also in: A.M.\ Vinogradov (ed.), Symmetries of
  partial differential equations. Conservation laws - Applications -
  Algorithms, Kluwer Acad.\ Publ., Dordrecht, 1989.

\bibitem{Mikhalev-1992} V. G. Mikhalev, \emph{On the Hamiltonian formalism for
    Korteweg-de Vries type hierarchies}, Funktsional. Anal. i Prilozhen.,
  \textbf{26:2} (1992), 79--82; Funct. Anal. Appl., \textbf{26:2} (1992),
  140--142.

\bibitem{Pavlov-2003} M.V.~Pavlov, \emph{Integrable hydrodynamic chains}, J.\
  Math.\ Phys., \textbf{44} (2003), 4134--4156.

\bibitem{Pavlov-2009} M.V.~Pavlov, Jen Hsu Chang, Yu Tung Chen, Integrability
  of the Manakov-Santini hierarchy, \url{arXiv:0910.2400}, 2009.
\end{thebibliography}
\end{document}